\documentclass{article}[11pt,a4paper]

\usepackage{amsmath}
\usepackage{graphicx}

\def \lket {|}
\def \rket {\rangle}

\def \A {{\cal A}}
\def \B {{\cal B}}

\def \R {{\cal R}}

\def\bbbr{R}
\newcommand{\ket}[1]{\lket #1\rket}

\newcommand{\comment}[1]{}

\newtheorem{Theorem}{Theorem}
\newtheorem{Lemma}{Lemma}
\newtheorem{Corollary}{Corollary}
\newtheorem{Claim}{Claim}
\newcommand{\proof}{\noindent {\bf Proof: }}
\newcommand{\qed}{\nobreak \ifvmode \relax \else
      \ifdim\lastskip<1.5em \hskip-\lastskip
      \hskip1.5em plus0em minus0.5em \fi \nobreak
      \vrule height0.75em width0.5em depth0.25em\fi}

\begin{document}

\title{Worst case analysis of non-local games
\thanks{Supported by ESF project 2009/0216/1DP/1.1.1.2.0/09/APIA/VIAA/044,
FP7 Marie Curie International Reintegration
Grant PIRG02-GA-2007-224886 and FP7 FET-Open project QCS.}
}
\author{Andris Ambainis,
Art\=urs Ba\v ckurs,
Kaspars Balodis, \\
Agnis \v Sku\v skovniks,
Juris Smotrovs,
Madars Virza\\ \\
Faculty of Computing, University of Latvia, \\
Raina bulv. 19, Riga, LV-1586, Latvia
}
\date{Keywords: quantum computing, non-local games}

\maketitle

\abstract{
Non-local games are studied in quantum information because they provide a simple way for proving 
the difference between the classical world and the quantum world.
A non-local game is a cooperative game played by 2 or more players against a referee. 
The players cannot communicate but may share common random bits or a common quantum state.
A referee sends an input $x_i$ to the $i^{\rm th}$ player who then responds by sending an
answer $a_i$ to the referee. The players win if the answers $a_i$ satisfy a condition that
may depend on the inputs $x_i$.

Typically, non-local games are studied in a framework where the referee picks the inputs from
a known probability distribution. We initiate the study of non-local games in 
a worst-case scenario when the referee's probability distribution is unknown and study 
several non-local games in this scenario.
}

\section{Overview}

Quantum mechanics is strikingly different from classical physics. 
In the area of information processing, this difference can be seen through quantum
algorithms which can be exponentially faster than conventional algorithms \cite{Simon,Shor} and 
through quantum cryptography which offers degree of security that is impossible classically \cite{BB84}.

Another information-theoretic way of seeing the difference between 
quantum mechanics and the classical world is through non-local games. 
An example of non-local game is the CHSH (Clauser-Horne-Shimonyi-Holt) game \cite{CHSH}. 
This is a game
played by two parties against a referee. The two parties cannot communicate but can 
share common randomness or common quantum state that is prepared before the beginning
of the game. The referee prepares two uniformly random bits $x, y$ and gives one
of them to each of two parties. The parties reply by sending bits $a$ and $b$ 
to the referee. They win if $a\oplus b = x\land y$. The maximum winning probability
that can be achieved is 0.75 classically and $\frac{1}{2}+\frac{1}{2\sqrt{2}} = 0.85...$
quantumly. 

Other non-local games can be obtained by changing the winning conditions, replacing bits $x, y$ with values $x, y\in\{1, \ldots, m\}$ or changing the number of parties.
The common feature is that all non-local games involve parties that cannot communicate
but can share common random bits or common quantum states.

There are several reasons why non-local games are interesting. First, CHSH game provides
a very simple example to test validity of quantum mechanics. 
If we have implemented the referee and the two players by devices so that there is no communication possible between $A$ and $B$ and we observe the winning probability of 0.85..., there is no classical explanation possible. 
Second, non-local games have been used in device-independent cryptography.

Non-local games are typically analyzed with the referee acting according to some probability
distribution. For example, in the case of the CHSH game, the referee chooses each of possible
pairs of bits $(0, 0)$, $(0, 1)$, $(1, 0)$, $(1, 1)$ as $(x, y)$ with equal probabilities
1/4. 
In this paper, we initiate study of non-local games in a worst case setting,
when the referee's probability distribution is unknown and the players have to achieve
winning probability at least $p$ for every possible input $(x, y)$.

We analyze a number of games in the worst-case framework. For some of them, the worst-case
winning probability turns out to be the same as the winning probability under the typically studied probability distributions. 
For example, for the CHSH game, the worst-case winning probability is the same 0.75 classically and $\frac{1}{2}+\frac{1}{2\sqrt{2}} = 0.85...$
quantumly.    

\section{Technical preliminaries}

We will study non-local games of the following kind \cite{Cleve} in both classical and quantum settings.
There are $n$ players $A_1$, $A_2$, \dots, $A_n$ which cooperate between themselves
to maximize the game value (see below), and there is a referee.
Before the game the players may share a common source of correlated random
data: in the classical case, a common random variable $R$ taking values in a finite set 
$\R$, and in the quantum case, an entangled $n$-part quantum state 
$\ket\psi\in\A_1\otimes\ldots\otimes\A_n$ (where $\A_i$ is a finite-dimensional subspace
corresponding to the part of the state available to the player $A_i$). 
During the game the players cannot communicate between themselves.

Each of the players ($A_i$) has a finite set of possible input data: $X_i$.
At the start of the game the referee randomly picks values 
$(x_1,\ldots,x_n)={\bf x}\in X_1\times\ldots\times X_n$ according to some probability
distribution $\pi$, and sends each of the players his input (i.~e. $A_i$ receives $x_i$).

Each of the players then must send the referee a response $a_i$ which may
depend on the input and the common random data source. 
(Any additional, local
randomization the player could employ can be technically incorporated in
the random variable $R$, so we will disregard it.)
In this paper we will consider only {\em binary games}, that is games
where the responses are simply bits: $a_i\in\{0,1\}$. We denote
$(a_1,\ldots,a_n)$ by ${\bf a}$.

The referee checks whether the players have won by some predicate
(known to all parties) depending on the players' inputs and outputs: 
$V({\bf a}\mid {\bf x})$. For convenience in formulas, we will
suppose that $V$ takes value $1$ when it is {\em true} and $-1$ when it is {\em false}.
A binary game whose outcome actually depends only on the $XOR$ of the players'
responses: $V({\bf a}\mid {\bf x})=V^\prime(\bigoplus_{i=1}^na_i\mid {\bf x})$, is called an {\em XOR game}.
A game for which the outcome does not change after any permutation
of the players (i.~e. $V(\gamma({\bf a})\mid\gamma({\bf x}))=V({\bf a}\mid {\bf x})$
for any permutation $\gamma$) is called a {\em symmetric game}.

The value $\omega$ of a non-local game $G$ for given strategies of the players
is the difference between the probability that the players win and the probability 
that they lose:
$$
\omega(G)=\Pr[V({\bf a}\mid{\bf x})=1]-\Pr[V({\bf a}\mid{\bf x})=-1]\in[-1,1].
$$
The probability that the players win can then be expressed by the game value
in this way: $\Pr[V({\bf a}\mid{\bf x})=1]=\frac12+\frac12\omega(G)$.

In the classical case, the
players' strategy is the random variable $R$ and a set of functions $a_i:
\;X_i\times\R\to\{0,1\}$ determining the responses. 
The maximal classical game value achievable by the players for a given probability distribution $\pi$ is thus:
$$
\omega^\pi_c(G)=\sup_{R,{\bf a}}\sum_{r,{\bf x}}\pi({\bf x})\Pr[R=r]V(a_1(x_1,r),\ldots,a_n(x_n,r)\mid{\bf x}).
$$
However, actually the use of random variable here is redundant, since in the
expression it provides a convex combination of deterministic strategy game values,
thus the maximum is achieved by some deterministic strategy (with $a_i:\;X_i\to\{0,1\}$):
$$
\omega^\pi_c(G)=\max_{{\bf a}}\sum_{{\bf x}}\pi({\bf x})V(a_1(x_1),\ldots,a_n(x_n)\mid{\bf x}).
$$

In this paper we investigate the case when the players
do not know the input values probability distribution $\pi$ used by the referee,
and must maximize the game value for the worst distribution $\pi$ that the
referee could choose, given the strategy picked by the players. We will call
it the worst-case game value. The maximal classical worst-case game value $\omega_c$ 
achievable by the players is given by the formula
$$
\omega_c(G)=\sup_{R,{\bf a}}\min_\pi\sum_{r,{\bf x}}\pi({\bf x})\Pr[R=r]V(a_1(x_1,r),\ldots,a_n(x_n,r)\mid{\bf x}).
$$
Note that in the worst-case approach the optimal strategy cannot be a
deterministic one, unless there is a deterministic strategy winning on {\em all}
inputs: if there is an input on which the strategy loses, then the referee
can supply it with certainty, and the players always lose.
Clearly, $\omega_c(G)\leq \omega^\pi_c(G)$ for any $\pi$.

In the most of the studied examples $\pi$ has been the uniform distribution.
We will call it the average case and denote its maximum game value by 
$\omega_c^{\rm uni}(G)$.

In the quantum case, the players' strategy is the state $\ket\psi$ and the measurements
that the players pick depending on the received inputs and perform on their parts
of $\ket\psi$ to determine their responses. Mathematically, the measurement performed
by $A_i$ after receiving input $x_i$ is a pair of positive semidefinite 
$\dim\A_i$-dimensional matrices $M_i^{0\mid x_i}$, $M_i^{1\mid x_i}$ with
$M_i^{0\mid x_i}+M_i^{1\mid x_i}=I$ where $I$ is the identity matrix.
We denote the collection of all measurements by ${\bf M}$.

The maximum quantum game value for a fixed distribution $\pi$ is
$$
\omega^\pi_q(G)=\sup_{\ket\psi,{\bf M}}\sum_{{\bf x},{\bf a}}\pi({\bf x})\langle\psi|\bigotimes_{i=1}^nM_i^{a_i\mid x_i}|\psi\rangle V({\bf a}\mid{\bf x}),
$$
and the maximum quantum worst-case game value is
$$
\omega_q(G)=\sup_{\ket\psi,{\bf M}}\min_\pi\sum_{{\bf x},{\bf a}}\pi({\bf x})\langle\psi|\bigotimes_{i=1}^nM_i^{a_i\mid x_i}|\psi\rangle V({\bf a}\mid{\bf x}).
$$

Since the shared entangled state can be used to simulate a random variable,
$\omega_q(G)\geq\omega_c(G)$ and for any $\pi$: $\omega^\pi_q(G)\geq\omega^\pi_c(G)$.

In the case of two player games ($n=2$) we will use notation $A,B$ for the players,
$X,Y$ for the input sets, $x,y$ for the inputs, $a,b$ for the responses, $\A,\B$
for the players' subspaces.


\section{Games with worst case equal to average case}

For several commonly studied non-local games, the worst case and the average case game values are the same:
$\omega_c^{uni}=\omega_c$ and $\omega_q^{uni}=\omega_q$.
The reason for that is the natural symmetries of the non-local games.
These symmetries often result in natural strategies which achieve the 
game value $\omega_c^{uni}$ or $\omega_q^{uni}$ on all inputs simultaneously, 
thus proving that $\omega_c^{uni}=\omega_c$ or $\omega_q^{uni}=\omega_q$.

We show that this happens for two well-known non-local games: CHSH game (which is a canonical example
of a 2-player non-local game with a quantum advantage) and Mermin-Ardehali game (which is the n-player XOR game
with the biggest advantage for quantum strategies). Other examples with $\omega_c^{uni}=\omega_c$ and $\omega_q^{uni}=\omega_q$
(for similar reasons) are the Odd Cycle game of \cite{Cleve} and the Magic Square game of \cite{Arvind,Cleve}.

\subsection{CHSH game}

The CHSH game is a two player XOR game with $X=Y=\{0,1\}$, 
$V(a,b\mid x,y)=a\oplus b\equiv x\land y$, and $\pi$ the uniform distribution.
It is easy to check that no deterministic strategy can win on all inputs, 
but the strategy $a(x)=0$, $b(y)=0$ wins on 3 inputs out of 4, so \cite{Cleve}:
$\omega_c^{\rm uni}(CHSH)=0.75-0.25=0.5$.

Moreover, since out of the four strategies $S_1$: $a(x)=0$, $b(y)=0$; 
$S_2$: $a(x)=x$, $b(y)=0$; 
$S_3$: $a(x)=0$, $b(y)=y$; $S_4$: $a(x)=x$, $b(y)=\neg y$ each one loses on a
different input, and wins on the 3 other ones, we have for any predetermined $\pi$:
$\omega^\pi_c(CHSH)=1-2\min_{x,y}\pi(x,y)\geq 0.5$. Indeed, one can pick the strategy
losing on the input with the minimal value of $\pi$.

\begin{Theorem}
$\omega_c(CHSH)=0.5$.
\end{Theorem}
\proof
If the players use the random variable $R$ to pick one of the atrategies
$S_1$, $S_2$, $S_3$, $S_4$ mentioned above with equal probability (i.~e. 0.25), 
then for any input $x,y$ they will have a winning strategy with probability 0.75.
Thus $\omega_c(CHSH)\geq 0.5$. On the other hand, 
$\omega_c(CHSH)\leq \omega^{\rm uni}_c(CHSH)=0.5$.
\qed

\cite{Cleve} shows that the winning probability in the quantum case
is $\frac12+\frac1{2\sqrt2}$ giving $\omega_q^{\rm uni}(CHSH)=1/\sqrt{2}$.
Moreover, the used strategy achieves this value {\em on every input}
$x,y$, therefore it gives also the worst-case value:

\begin{Theorem}
$\omega_q(CHSH)=1/\sqrt{2}$.
\end{Theorem}

\subsection{Mermin-Ardehali game}


Mermin-Ardehali (MA) game \cite{Mermin,Ardehali} 
is an $n$-player XOR game with $X_1=\ldots=X_n=\{0,1\}$ and the winning condition:
$a_1\oplus \ldots\oplus a_n = 0$ if 
$(x_1+\ldots+x_n) \bmod 4 \in \{0, 1\}$ and
$a_1\oplus \ldots\oplus a_n = 1$ if 
$(x_1+\ldots+x_n) \bmod 4 \in \{2, 3\}$.

For the uniform distribution on the inputs, 
this game can be won with probability $\frac{1}{2}+\frac{1}{2\sqrt{2}}$ quantumly
and $\frac{1}{2}+\frac{1}{2^{\lfloor \frac{n}{2}+1 \rfloor}}$ classically \cite{Mermin,Ardehali,AKNR}.
Thus, the ratio between its quantum and classical values, 
$\frac{\omega^{\rm uni}_{q}(MA)}{\omega^{\rm uni}_{c}(MA)}$, 
is equal to $2^{\lfloor \frac{n}{2} \rfloor -\frac{1}{2} }$.

This game corresponds to Mermin-Ardehali $n$-partite Bell inequality \cite{Mermin,Ardehali}.
For even $n$, this inequality is the Bell inequality that can be violated 
by the biggest possible margin in the quantum theory \cite{WW01}.

{\sloppy
Namely, Werner and Wolf \cite{WW01} have shown the following theorem (translated
here from the language of Bell inequalities to the language of XOR games):

}
\begin{Theorem}
\label{thm:ww}
\cite{WW01}
No n-party XOR game $G$ with binary inputs $x_i$ (with any input distribution $\pi$)
achieves $\frac{\omega^\pi_{q}(G)}{\omega^\pi_{c}(G)}>2^{\frac{n-1}{2}}$.
\end{Theorem}

This makes the worst-case analysis of Mermin-Ardehali game for even $n$ quite straightforward.
For the quantum case, the maximal game value $\frac{1}{\sqrt{2}}$ is
given by a quantum strategy which achieves it {\em on every input}
\cite{Mermin,Ardehali,AKNR}, thus
\begin{Theorem}
For all $n$: $\omega_q(MA)=1/\sqrt{2}$.
\end{Theorem}

For the classical case, the worst-case game value cannot be better than
the game value $2^{-\frac{n}{2}}$ for the uniform distribution. 
If the worst-case value was 
$\omega_c(MA)<2^{-\frac{n}{2}}$ then, by Yao's principle \cite{Yao}, there would be 
a specific probability distribution $\pi$ for 
which no classical strategy achieves value exceeding $\omega_c(MA)$. 
Then, we would obtain a ratio 
$\frac{\omega^\pi_{q}(MA)}{\omega^\pi_{c}(MA)}>2^{\frac{n-1}{2}}$ 
(because the quantum game value would still be at least 
$1/\sqrt{2}$),
contrary to Theorem \ref{thm:ww}. Hence
\begin{Theorem}
For even $n$: $\omega_c(MA)=2^{-\frac{n}{2}}$.
\end{Theorem}

By similar reasoning, $2^{-\frac{n}{2}}$ is a lower bound for the classical 
worst-case game value for odd $n$. However, the best upper bound that we can prove in this case is 
$2^{-\frac{n-1}{2}}$ (because this is the value for the uniform
distribution).


\section{Games with worst case different from average case}

In the previous section, we saw that, for many well-known
non-local games, the worst case probability distribution is the uniform distribution. 
We now present several games for which this is not the case.

\subsection{EQUAL-EQUAL game}


We define EQUAL-EQUAL ($EE_m$) as 
a two-player XOR game with $X=Y=\{1,\ldots,m\}$ and 
$V(a,b\mid x,y)\stackrel{def}{=}(x=y)\equiv(a=b)$.

This is a natural variation of the Odd-Cycle game of \cite{Cleve}.
For $m=3$, the Odd-Cycle game can be viewed as a game in which 
the players attempt to prove to the referee that they have 3 bits 
$a_1, a_2, a_3\in\{0, 1\}$ which all have different values. 

This can be generalized to larger $m$ in two different ways.
The first generalization is the Odd-Cycle game \cite{Cleve} in which the players attempt to 
prove to the referee that an $m$-cycle (for $m$ odd) is 2-colorable.
The second generalization is a game in which the players attempt to prove that
they have $m$ bits $a_1, \ldots, a_m\in\{0, 1\}$ which all have different values.
This is our EQUAL-EQUAL game.

\begin{Theorem}
\label{cEEm}
For even $m$: $\omega_c(EE_m)=\frac{m}{3m-4}$, and
for odd $m$: $\omega_c(EE_m)=\frac{m+1}{3m-1}$.
\end{Theorem}

\proof
For the lower bound, we provide a strategy achieving the needed game value
on all inputs. It uses the random variable $R$ to obtain the following
probabilistic mix of deterministic strategies: with some probability
$p$ the strategy having $a=0$ and $b=1$ on all inputs is picked,
and with probability $1-p$ a strategy of the following kind is chosen
uniformly randomly: for some fixed $\lfloor m/2\rfloor$ input values $i$: 
$a(i)=b(i)=0$, and for the remaining $\lceil m/2\rceil$ values $i$: $a(i)=b(i)=1$.

The first strategy wins with certainty if $x\neq y$, and loses if $x=y$.
The strategy mix of the other case wins with certainty if $x=y$.
In the case $x\neq y$ a particular strategy of the mix loses
iff $a(x)=b(y)$. There are $\lfloor(m-1)^2/2\rfloor$ such
input pairs among the $m(m-1)$ pairs with $x\neq y$.
Picking the strategy randomly, for any such input pair we get
that the probability of losing is $\frac{\lfloor(m-1)^2/2\rfloor}{m(m-1)}$.

Easy calculations show that by picking $p=\frac{m-2}{3m-4}$ for even $m$
and $p=\frac{m-1}{3m-1}$ for odd $m$, we obtain the desired game value
for both cases $x=y$ and $x\neq y$.

For the upper bound, we provide a probability distribution $\pi$
for which no deterministic strategy can exceed the game value
of the theorem's statement. Let us denote by $\pi_{\alpha,\beta}$
a distribution having $\pi_{\alpha,\beta}(i,i)=\alpha$ for any $i$ and
$\pi_{\alpha,\beta}(i,j)=\beta$ for any distinct $i$, $j$. We will use $\pi_{\alpha,\beta}$
with 
\begin{equation}
\label{cEEm:eq1}
\alpha=
\begin{cases}
\frac{m-1}{m(3m-1)} & \text{if $m$ is odd,}\\
\frac{m-2}{m(3m-4)} & \text{if $m$ is even,}
\end{cases}
\quad\beta=
\begin{cases}
\frac{2}{(m-1)(3m-1)} & \text{if $m$ is odd,}\\
\frac{2}{m(3m-4)} & \text{if $m$ is even.}
\end{cases} 
\end{equation}

Consider a deterministic strategy. The input values $i$ are split into
two classes: the ones with $a(i)\neq b(i)$ and the ones with $a(i)=b(i)$.
We denote the number of the elements of the first class with $k$.
Let us estimate from below the probability of loss $p_{\rm loss}$. 

The cases when $x=y$ and both $x,y$ belong to the first class 
contribute $k\alpha$ to $p_{\rm loss}$.

If $i$ belongs to the first class, and $j$ to the second, then the
strategy loses in exactly one of the cases $x=i$, $y=j$ and $x=j$, $y=i$,
since exactly one of the different values $a(i)$, $b(i)$
will coincide with $a(j)=b(j)$. That contributes $k(m-k)\beta$ to $p_{\rm loss}$.

Finally, if both $x$ and $y$ are from the second class,
then $a(x)=b(x)$ and $a(y)=b(y)$, and the strategy loses if $x\neq y$ and $a(x)=a(y)$.
It is easy to check that the minimum number of loss cases 
$\lfloor(m-k-1)^2/2\rfloor$ is achieved when the number of values
$i$ with $a(i)=b(i)=0$ is the nearest possible to one half of the elements
of the class. The contribution to $p_{\rm loss}$ in this case is 
$\lfloor(m-k-1)^2/2\rfloor\beta$. 
Thus we have 
$p_{\rm loss}\geq k\alpha+k(m-k)\beta+\lfloor(m-k-1)^2/2\rfloor\beta$.

In all the parity cases of $m$ and $k$ this is a concave (quadratic) function of $k$,
so it achieves minimum at the minimal or maximal possible value of $k$
(e.~g. if $k$ and $m$ are odd, then $k=1$ or $k=m$).
A routine checking shows that in all these cases the desired upper bound 
of the game value is obtained.
\qed

\begin{Theorem}
\label{qEEm}
For even $m$: $\omega_q(EE_m)=\frac{m}{3m-4}$, and
for odd $m$: $\frac{m+1}{3m-1}\leq\omega_q(EE_m)\leq\frac{m^2+1}{(3m-1)(m-1)}$.
\end{Theorem}

\proof
The lower bounds follow from $\omega_q(EE_m)\geq\omega_c(EE_m)$.

To obtain the upper bounds, we will take the probability distribution
$\pi_{\alpha,\beta}$ from the proof of the previous theorem and use 
$\omega_q(EE_m)\leq\omega^{\pi_{\alpha,\beta}}_q(EE_m)$.

For the two-player XOR games where on every input exactly one
of the cases $a\oplus b=0$ and $a\oplus b=1$ is winning, it is useful
to observe that $V(a,b\mid x,y)=(-1)^a(-1)^bV(0,0\mid x,y)$, and
to introduce the matrix $V$ with $V_{xy}=V(0,0\mid x,y)$. 
Thus, for any distribution $\pi$
$$
\omega^{\pi}_q(EE_m)=\sup_{\ket\psi,{\bf M}}
\sum_{x,y,a,b}\pi(x,y)\langle\psi|M_1^{a\mid x}\otimes M_2^{b\mid y}|\psi\rangle 
(-1)^a(-1)^bV_{xy}.
$$

The Tsirelson's theorem \cite{T} implies that this game value is equal to
$$
\sup_d\max_{u_i: \|u_i\|=1}\max_{v_j: \|v_j\|=1}\sum_{i=1}^m\sum_{j=1}^m\pi(i,j)V_{ij}(u_i,v_j)
$$
where $u_1, \ldots, u_m, v_1, \ldots, v_m\in\bbbr^d$ and $(u_i,v_j)$ is the scalar product.

The part of the sum containing $u_i$ is 
$$
\sum_{j=1}^m\pi(i,j)V_{ij}(u_i,v_j)=\left(u_i,\sum_{j=1}^m\pi(i,j)V_{ij}v_j\right).
$$
To maximize the scalar product, $u_i$ must be the unit vector in the direction of
$\sum_{j=1}^m\pi(i,j)V_{ij}v_j$. 

For the EQUAL-EQUAL game and the distribution $\pi_{\alpha,\beta}$ we have 
$V_{ij}=1$ and $\pi_{\alpha,\beta}(i,j)=\alpha$ if $i=j$, 
$V_{ij}=-1$ and $\pi_{\alpha,\beta}(i,j)=\beta$ if $i\neq j$.
So we have to maximize the sum
$$
S=\sum_{i=1}^m\left\|\sum_{j=1}^m\pi_{\alpha,\beta}(i,j)V_{ij}v_j\right\|=
\sum_{i=1}^m\left\|\alpha v_i-\beta\sum_{j=1,j\neq i}^mv_j\right\|.
$$

Let us denote $s = \sum_{j=1}^{m} v_j$ and apply the inequality between 
the arithmetic and quadratic means (and use the fact that $\|v_i\|=1$):
\begin{align*}
S^2 & \leq m \sum_{i=1}^m \| \alpha v_i - \beta (s - v_i) \|^2 
= m \sum_{i=1}^m \| (\alpha+\beta) v_i - \beta s \|^2\\
&= m\left( \sum_{i=1}^m (\alpha+\beta)^2 \|v_i\|^2 - \sum_{i=1}^m 2 (\alpha+\beta)\beta ( v_i, s ) + \sum_{i=1}^m \beta^2 \|s\|^2 \right)\\
&= m \left( (\alpha+\beta)^2m - 2 (\alpha+\beta)\beta \left(\sum_{i=1}^m v_i, s \right) + m \beta^2 \|s\|^2 \right) \\
&= m ( (\alpha+\beta)^2m - 2 (\alpha+\beta)\beta \|s\|^2 + m \beta^2 \|s\|^2 ) \\
&= m^2(\alpha+\beta)^2 + \|s\|^2 m\beta(m \beta - 2 (\alpha+\beta)).
\end{align*}
With our values of $\alpha$ and $\beta$ (see equation (\ref{cEEm:eq1}))
one can calculate that the coefficient at $\|s\|^2$ is 0 for even $m$ and
$-\frac{4}{(m-1)^2(3m-1)^2}$ (negative) for odd $m$,
so dropping this summand and extracting the square root
we get $S\leq m(\alpha+\beta)$. Substituting the values of $\alpha$ and $\beta$
according to equation (\ref{cEEm:eq1}) we get the desired estimations.
\qed

It follows from this result that at the worst-case distribution for any even $m$
the quantum strategy cannot achieve any advantage over the classical
strategies (and for odd $m$ there is no difference asymptotically).
It was quite surprising for us. In fact, it can be proven for any
of the distributions $\pi_{\alpha,\beta}$.

\begin{Theorem}
\label{EEalfabeta}
If $\beta\geq\frac{2}{m(3m-4)}$ then
$$
\omega^{\pi_{\alpha,\beta}}_q(EE_m)=\omega^{\pi_{\alpha,\beta}}_c(EE_m)=2\beta(m-1)m-1.
$$
If $\beta < \frac{2}{m(3m-4)}$ then for even $m$: 
$$
\omega^{\pi_{\alpha,\beta}}_q(EE_m)=\omega^{\pi_{\alpha,\beta}}_c(EE_m)=1-\beta(m-2)m,
$$
and for odd $m$: 
$$
1-\beta(m-1)^2\leq\omega^{\pi_{\alpha,\beta}}_c(EE_m)\leq\omega^{\pi_{\alpha,\beta}}_q(EE_m)\leq 1-\beta(m-2)m.
$$
\end{Theorem}
\proof
For the classical game value, we use the strategies from the proof of
Theorem \ref{cEEm}. In the case $\beta < \frac2{m(3m-4)}$ we use the 
strategy mix described in that Theorem. Recall that it wins with
certainty if $x=y$ and wins with probability $1-\lfloor\frac{(m-1)^2}{2}\rfloor/(m(m-1))$
if $x\neq y$. Thus its game value is
$$
\alpha m+\beta(m-1)m\left(1-\frac{2\lfloor(m-1)^2/2\rfloor}{m(m-1)}\right).
$$
Since the sum of the probabilities of the distribution is $1$, we have
$m\alpha+m(m-1)\beta=1$ or $\alpha=(1-m(m-1)\beta)/m$.
Substituting this and simplifying, we get that the game value is 
$1-2\beta\lfloor\frac{(m-1)^2}{2}\rfloor$. This is equal to
$1-\beta(m-2)m$ for even $m$, and $1-\beta(m-1)^2$ for odd $m$.
That gives the needed lower bounds for $\omega^{\pi_{\alpha,\beta}}_c$
in this case.

In the case $\beta \geq \frac2{m(3m-4)}$ we use the 
strategy responding for all $i$ with $a(i)=0$ and $b(i)=1$.
Since it wins iff $x\neq y$, its game value is
$\beta(m-1)m-\alpha m$.
Substituting the expression of $\alpha$ by $\beta$ we get that the 
game value is $2\beta(m-1)m-1$,
so $\omega^{\pi_{\alpha,\beta}}_c\geq 2\beta(m-1)m-1$.

For the quantum game value we will use the notation and intermediary 
results from the proof of Theorem \ref{qEEm}.
We obtained there that $(\omega^{\pi_{\alpha,\beta}}_q)^2\leq 
m^2(\alpha+\beta)^2 + \|s\|^2 m\beta(m \beta - 2 (\alpha+\beta))$.
Substituting the $\alpha$ we get 
$$
(\omega^{\pi_{\alpha,\beta}}_q)^2\leq 
(1-\beta(m-2)m)^2 + \|s\|^2 \beta m(3m-4)\left(\beta - \frac2{m(3m-4)}\right).
$$

If $\beta - \frac2{m(3m-4)}<0$ then the second summand is non-positive,
so we can drop it obtaining 
$\omega^{\pi_{\alpha,\beta}}_q\leq 1-\beta(m-2)m$.

If $\beta - \frac2{m(3m-4)}\geq 0$ then the second summand is non-negative,
so we can estimate it from above by maximizing $\|s\|$. Since
$s$ is a sum of $m$ vectors of unit length, $\|s\|\leq m$.
Substituting this value and simplifying we get
$\omega^{\pi_{\alpha,\beta}}_q\leq 2\beta(m-1)m-1$.
\qed

\begin{Corollary}
For $m\geq 4$: $\omega^{\rm uni}_q(EE_m)=\omega^{\rm uni}_c(EE_m)=\frac{m-2}{m}$.
\end{Corollary}
\proof
The uniform distribution is $\pi_{\alpha,\beta}$ with $\alpha=\beta=1/m^2$.
Substituting these values in the previous theorem gives the result.
\qed


\subsection{n-party AND game}


$n$-party AND game ($nAND$) is a symmetric XOR game with binary inputs 
$X_1=\ldots=X_n=\{0,1\}$ and $V({\bf a}\mid{\bf x})=(\bigoplus_{i=1}^na_i=\bigwedge_{i=1}^nx_i)$.

Although this is a natural generalization of the CHSH game
(compare the winning conditions),
it appears that this game has not been studied before. Possibly, this is due to the fact that in the average case the
game can be won classically with a probability that is very close to 1 by a trivial strategy:
all players always outputting $a_i=0$. If this game is studied in the worst-case scenario,
it becomes more interesting. The following theorem shows that
$\lim_{n\to\infty}\omega_c(nAND)=1/3$.

\begin{Theorem}
\label{cnAND}
$\omega_c(nAND)=2^{n-2}/(3\cdot 2^{n-2}-1)$.
\end{Theorem}
\proof
First, let us notice that there are only four possible deterministic strategies
for any individual player $A_i$: $a_i(x_i)$ can be $0$, $1$, $x_i$ or $\neg x_i$.
We will denote a deterministic $n$-player strategy by the tuple
of individual strategies $(a_1(x_1),\ldots,a_n(x_n))$.

For the lower bound
we construct a probabilistic mix of deterministic
strategies which achieves the needed game value on all inputs:
we take the strategy $(0,0,\ldots,0)$ with probability 
$(2^{n-1}-1)/(3\cdot 2^{n-1}-2)$, and we take a strategy picked
uniformly randomly from the class desribed below with 
the remaining probability $(2^n-1)/(3\cdot 2^{n-1}-2)$.

This class consists of all the strategies where all the non-zero
elements $a_i(x_i)$ are $x_i$ except for the last non-zero element which
is either $x_i$ or $\neg x_i$ picked so that the total number of $x_i$'s
in the tuple is odd. It is easy to check that all these strategies
win on the input $x_1=\ldots=x_n=1$ (due to the odd number of $x_i$'s), 
and that there are $2^n-1$ such strategies (every element of the tuple 
can be either $0$ or non-zero; the all-zero tuple is excluded).

Furthermore, for any input apart from the all-ones input there are
$2^{n-1}$ strategies of the class winning on this input. Indeed, any such input
has at least one $0$, suppose $i$ is the position of the first $0$.
Then we can split all the strategies of the class except one in pairs:
if the strategy has $0$ in the $i$-th position, then we obtain its
match by exchanging the $0$ with $x_i$ and adjusting the negation
as appropriate, and vice versa, if the $i$-th element is non-zero,
we put there $0$ and adjust the negation. Checking cases
(regarding the position of negation), it is easy
to see that the strategies of one pair have different outputs
on the given input, so exactly half of them wins.
The one strategy without pair is $(0,\ldots,0,x_i,0,\ldots,0)$, and it wins.

The strategy $(0,0,\ldots,0)$ loses on the all-ones input, and wins on all the others.
So, as needed, the game value of our mix of strategies on the all-ones input is
$$
\frac{2^n-1}{3\cdot 2^{n-1}-2}-\frac{2^{n-1}-1}{3\cdot 2^{n-1}-2}=\frac{2^{n-2}}{3\cdot 2^{n-2}-1},
$$
and on all the other inputs
$$
\frac{2^{n-1}-1}{3\cdot 2^{n-1}-2}+\frac{2^{n-1}}{3\cdot 2^{n-1}-2}-\frac{2^{n-1}-1}{3\cdot 2^{n-1}-2}=\frac{2^{n-2}}{3\cdot 2^{n-2}-1}.
$$

For the upper bound we introduce a distribution $\pi$ for which no deterministic
strategy can exceed the desired value: $\pi(1,1,\ldots,1)=(2^{n-1}-1)/(3\cdot 2^{n-1}-2)$,
and for all the other inputs $\pi(x_1,\ldots,x_n)=1/(3\cdot 2^{n-1}-2)$.

Any constant strategy (i.~e. consisting only of $0$ and $1$) can do no better
than losing on the all-ones input and winning on all the other ones, and then
the game value is
$$
\frac{2^n-1}{3\cdot 2^{n-1}-2}-\frac{2^{n-1}-1}{3\cdot 2^{n-1}-2}=\frac{2^{n-2}}{3\cdot 2^{n-2}-1}.
$$

If the strategy is not constant, let $i$ be the index of its first non-constant
element ($x_i$ or $\neg x_i$). Then all the inputs are split into pairs
differing only in the $i$-th bit, and the strategy having different values
for the inputs of one pair, wins only in one of them, except for the pair
containing the all-ones input where it may win on both inputs, because it
is the only case when the responses {\em need} to be different. Thus it
can maximally achieve the game value
$$
\frac{2^{n-1}-1}{3\cdot 2^{n-1}-2}+\frac{1}{3\cdot 2^{n-1}-2}+\frac{2^{n-1}-1}{3\cdot 2^{n-1}-2}-\frac{2^{n-1}-1}{3\cdot 2^{n-1}-2}=\frac{2^{n-2}}{3\cdot 2^{n-2}-1}.
$$
\qed



In the quantum case, since the game is symmetric with binary inputs,
we can introduce parameters $c_i$ being equal to the value of
$V((0,\ldots,0)\mid{\bf x})$ on any input ${\bf x}$ containing
$i$ ones and $n-i$ zeroes, and $p_i$ being equal to the probability
(determined by $\pi$) of such kind of input.
According to \cite{AKNR}, for such game $G$:
\[ \omega^\pi_q(G)=\max_{z:|z|=1} \left| \sum_{i=0}^n p_i c_i z^i \right| \]
where $z$ is a complex number.
By Yao's principle,
\[ \omega_q(G)=\min_{p_0, \ldots, p_n:\sum p_i=1}\;
\max_{z:|z|=1} 
\left| \sum_{i=0}^n p_i c_i z^i \right| .\]
We have for the $nAND$ game: $c_0=\ldots=c_{n-1}=1$ and $c_n=-1$.
The following lemma implies that, for large $n$, the quantum value of the game
is $\frac{1}{3}+o(1)$ and, hence, the maximum winning probability that can be achieved
is $\frac{2}{3}+o(1)$.

\begin{Lemma}
\label{lem:ineq}
$$
\lim_{n \to \infty}
\;
\min_{p_0, ..., p_n:\sum p_i=1}
\;
\max_{z:|z|=1}
\left|
\sum_{i=0}^{n-1}p_i z^i-p_nz^n
\right| \leq \frac{1}{3}
$$
\end{Lemma}

\proof
We prove the lemma by picking particular values of $p_i$
and showing for them that the limit is equal to $1/3$.
We take $p_n=1/3$, $p_i=pq^{n-i}$ for $i=0,\ldots,n-1$
where $q=e^{-\frac{1}{\sqrt{n}}}$ and $p$ is chosen so that $p\sum_{i=1}^{n}q^i=\frac{2}{3}$, i.e. $p=\frac23\frac{1-q}{q(1-q^n)}$.
Additionally, since $|z|=1$, we can divide the expression within
modulus by $z^n$ and use the substitution $w=1/z$. We obtain
\begin{equation}
\label{lem1:eq2}
\lim_{n \to \infty}
\max_{w:|w|=1}
\left|
p\sum_{i=1}^{n}(qw)^i-\frac13
\right|=
\lim_{n \to \infty}
\max_{w:|w|=1}
\left|
\frac23\frac{1-q}{1-q^n}\frac{w(1-q^nw^n)}{1-qw}-\frac13
\right|.
\end{equation}

From $\lim_{n \to \infty}q^n=\lim_{n \to \infty}e^{-\sqrt{n}}=0$ we get $\lim_{n \to \infty}(1-q^n)=1$ and, since $|w|=1$, $\lim_{n \to \infty}(1-q^nw^n)=1$.
Thus \eqref{lem1:eq2} is equal to
\begin{equation}
\label{lem1:eq3}
\lim_{n \to \infty}
\max_{w:|w|=1}
\left|
\frac23\frac{(1-q)w}{1-qw}-\frac13
\right|.
\end{equation}

\begin{Claim}
\label{claim:ineq}
For each $\epsilon>0$ there exists $\delta_0$ such that the inequality 
\begin{equation}
\label{claim1:eq1}
\left|\left|\frac{2\delta w}{1-(1-\delta)w}-1\right|-1\right|<\epsilon
\end{equation}
holds where $0<\delta<\delta_0$ and $z\in C$, and $|w|=1$.
\end{Claim}

Now Claim \ref{claim:ineq} gives that \eqref{lem1:eq3} is equal to $1/3$. We used the fact that $\lim_{n \to \infty}e^{-\frac{1}{\sqrt{n}}}=1$ and the substitution $1-q=\delta$.
\qed

\proof [of Claim \ref{claim:ineq}]

The inequality \eqref{claim1:eq1} requires that there exists some number with absolute value $1$ that is sufficiently close to $\frac{2\delta w}{1-(1-\delta)w}-1$ or, equivalently, that there exists some number on a circle in the complex plane with its center at $1/2$ and a radius of $1/2$ that is sufficiently close to $\frac{\delta w}{1-(1-\delta)w}=\frac{1}{1+((1 / w)-1)/\delta}$.

The numbers $\left\{\frac{1}{1+((1 / w)-1)/\delta}|w\in C \textnormal{ and }|w|=1\right\}$ form a circle in the complex plane with its center on the real axis that has common points with the real axis at $1$ and $\frac{1}{1-{2 / \delta}}=\frac{\delta}{\delta-2}$.
The latter circle is sufficiently close to the circle with its center at $1/2$ and radius of $1/2$ if we choose $\delta_0>0$ sufficiently small so that the value of $\frac{\delta}{\delta-2}$ is sufficiently close to $0$. 
\qed


\subsection{n-party MAJORITY game}

By replacing the AND function with the MAJORITY function in the definition
of the $n$-party AND game, we obtain the $n$-party MAJORITY game.

More formally, $n$-party MAJORITY game ($nMAJ$) is a symmetric XOR game with
$X_1=\ldots=X_n=\{0,1\}$ and 
$V({\bf a}\mid{\bf x})$ demanding that $\bigoplus_{i=1}^na_i$
is $true$ if at least half of $x_i$ is $true$, and $false$ otherwise.
Similarly as in the previous section, we introduce parameters $c_i$
and $p_i$ and use the expression for game value given in \cite{AKNR}. 
This time $c_0=\ldots=c_{\lceil n/2\rceil-1}=1$,
$c_{\lceil n/2\rceil}=\ldots=c_n=-1$. 

The following lemma implies that, for large $n$, the quantum (and thus also
the classical) value of the game is $o(1)$ and, hence, the maximum winning 
probability that can be achieved is $\frac{1}{2}+o(1)$. For odd $n$, set
$n=2k-1$. For even $n$, set $n=2k$, $p_n=0$ and use the lemma as an upper bound.

\begin{Lemma}
\label{lem:zero}
$\lim_{k \to \infty}f(k)= 0$
where
$$
f(k)=
\min_{p_0, p_1, ..., p_{2k-1}:\sum p_i =1}\;
\max_{z:|z|=1}
\left| 
\sum_{i=0}^{k-1}p_i z^i-\sum_{i=k}^{2k-1}p_i z^i
\right|.
$$
\end{Lemma}

\proof

It is clear that $f(k)\leq g(k)$ where we obtain $g$ from $f$ by substituting
particular values for $p_i$: $p_i=\frac{r_i}{s}$ where $r_i=r_{2k-1-i}$ and $r_i=\frac{1}{2k-1-2i}$ for $0 \leq i \leq k-1$, and $s=2\sum_{i=1}^k \frac{1}{2i-1}$.
We prove the lemma by showing that $\lim_{k \to \infty}g(k)=0$.

Since $|z|=1$, we can multiply the polynomial within the modulus by $z^{\frac12-k}$ and use the substitution $w=z^{-\frac12}$ obtaining that $g(k)$ is equal to:
$$
\max_{z:|z|=1}
\left| 
\sum_{i=0}^{k-1}p_i z^i-\sum_{i=k}^{2k-1}p_i z^i
\right|
=
\max_{w:|w|=1}
\left| 
\sum_{i=0}^{k-1}p_i w^{2k-1-2i}-\sum_{i=k}^{2k-1}p_i w^{2k-1-2i}
\right|
$$
$$
=
\frac2s
\max_{w:|w|=1}
\left|
\operatorname{Im}
\left(
\sum_{i=0}^{k-1}r_i w^{2k-1-2i}
\right)
\right|
=\frac2s\max_{\theta}\left| \sum_{i=1}^k \frac{\sin(2i-1)\theta}{2i-1}\right|
$$
where $\operatorname{Im}(z)$ is the imaginary part of $z$ and $w=e^{i\theta}$.

Since the function $\sum_{i=1}^k \frac{\sin(2i-1)\theta}{2i-1}$ is a partial sum of the Fourier series of a square wave function, we have $\max_{\theta}\left| \sum_{i=1}^k \frac{\sin(2i-1)\theta}{2i-1}\right|=O(1)$. Also, $\frac2s=o(1)$ because $\lim_{k \to \infty}s=\infty$. The result follows.
\qed


\section{Games without common data}


Finally, we consider the question: what can the players do if they are not allowed to share common randomness (nor common quantum state)? For the case when the probability distribution on the inputs is fixed, this
scenario is equivalent to two players who can share common randomness because common randomness can be always fixed to the value that achieves the best result for the two players.

In the worst-case setting, we get different results. For many games, not allowing shared randomness results
in the players being unable to win the game with any probability $p>1/2$. But for at least one game, players can still win with a non-trivial probability, even if they are not allowed to share randomness.

We will use the $\hat\omega$ notation for the game value in this case 
to distinguish it from the case with shared randomness.

\begin{Theorem}
Suppose $G$ is a two-player XOR game where on every input exactly
one of the two possible values of $a\oplus b$ wins.
If $\hat\omega_c(G)>0$ then $\hat\omega_c(G)=1$, i.~e. then
$G$ can be won deterministically.
\end{Theorem}
\begin{proof}
The probability to give the correct answer on input $\left( x , y \right)$ is either $$p_{1x} p_{2y} + \left( 1-p_{1x} \right) \left( 1-p_{2y} \right)\quad\mbox{or}\quad
p_{1x}  \left( 1 -p_{2y} \right) + \left( 1-p_{1x} \right) p_{2y}$$ 
where $p_{ij}$ is the probability that $i$-th player will give output $1$ on input $j$.
We can denote both cases as $p q + \left( 1 - p \right) \left( 1 - q \right) $ where $0 \leq p, q \leq 1$.
If $\hat\omega_c(G)>0$ then this expression must be greater than $\frac{1}{2}$,
so either both $p$ and $q$
or both $1-p$ and $1-q$ are greater than $\frac{1}{2}$ (if, say, $p>\frac{1}{2}$
and $q<\frac{1}{2}$, then replacing $q$ by $\frac{1}{2}$ would increase
the expression, and it would become equal to $\frac{1}{2}$, contrary to
the observation that initially it must have exceeded $\frac{1}{2}$).
If we increase them to $1$, the value of the expression only increases.

So by increasing all the probabilities exceeding $\frac{1}{2}$ to $1$ and decreasing the others to $0$ we get a deterministic strategy that always wins.
\qed
\end{proof}

\begin{Theorem}
There exists a two-player bit game $G$ with 
$0<\hat\omega_c(G)=\left( \sqrt{5} - 2 \right)<1$.
\end{Theorem}
\begin{proof}
Consider a game with $X=Y=\{0,1\}$ and $V(a,b\mid x,y)=x \vee y \equiv a \wedge b$.

Let $p_{ij}$ be the probability that the $i$-th player on input $j$ gives output $1$.
If the player gets input $1$, it's always better to give output $1$ than $0$,
so $p_{11}=p_{21}=1$. Therefore the probability of players winning on different inputs are as follows:
$$
\begin{array}{ll}
00: 1 - p_{00} p_{10} & 10: p_{10}\\
01: p_{00} & 11: 1
\end{array}
$$
At the maximal probability of giving the correct answer on the worst input we have $p_{00} = p_{10}$ (if one of them would be less than the other, we could increase it).
Let's denote this value by $p$.
Then the best result is achieved when $1- p^2 = p$.
The only positive solution is $p = \frac{1}{2}\left( \sqrt{5} - 1 \right)$.
The result follows.
\qed
\end{proof}


\end{document}